\documentclass[a4paper,11pt]{article}

\usepackage{amssymb}
\usepackage {amssymb,latexsym,amsthm,amsmath,enumerate,amscd,epsfig,wasysym,textcomp}

\usepackage{algorithmic}
\usepackage{amsmath}
\usepackage{pbox}

\usepackage{mathptmx}       
\usepackage{helvet}         
\usepackage{courier}        
\usepackage{makeidx}         
\usepackage{graphicx}        
\usepackage{multicol}        

\usepackage{appendix}
\makeindex             
\newtheorem{theorem}{Theorem}
\newtheorem{example}{Example}
\newtheorem{lemma}{Lemma}
\title{Improved Strength Four Covering Arrays with Three Symbols }
\author{Soumen Maity$^a$ $~~~~$ Yasmeen Akhtar$^a$\\ 
Reshma C. Chandrasekharan$^a$ $~~~~$ Charles J. Colbourn$^b$  \\
\it\small$^a$Indian Institute of Science Education and Research Pune, India\\
\it \small$^b$School of Computing, Informatics and Decision Systems Engineering,\\
\it \small Arizona State University, U.S.A
}
\date{}
\begin{document}
\maketitle

\begin{abstract}
A covering array $t$-$CA(n,k,g)$, of size $n$, strength $t$, degree $k$, and order $g$, is a $k\times n$ array 
on $g$ symbols such that   every $t\times n$ sub-array contains every $t\times 1$ column on $g$ symbols  at least once. 
Covering arrays have been studied for their applications to software testing, hardware testing, drug screening, and in areas 
where interactions of multiple parameters are to be tested. 
In this paper, we present an algebraic construction that improves many of the best known upper bounds on $n$ for 
covering arrays 4-$CA(n,k,g)$ with $g=3$. 
The {\it coverage measure} $\mu_t(A)$ of a testing array $A$ 
is defined by the ratio between the number of distinct $t$-tuples  contained in the column vectors of $A$ and the total number of 
$t$-tuples. A covering array is a testing array with full coverage.  The {\it covering arrays with budget constraints problem} is the problem of constructing 
a testing array of size at most $n$ having largest possible coverage measure, given values of $k,g$ and $n$.  This paper presents  several strength four testing arrays with 
high coverage.  The construction here is a generalisation of the construction 
methods used by Chateauneuf, Colbourn and   Kreher,  and Meagher and Stevens.

\end{abstract}







\section{Introduction}
 This article focuses on constructing  new strength-four covering arrays with $g=3$ and establishing  improved bounds on the covering array numbers 4-$CAN(k,3)$. 
 This article also presents solution to the covering arrays with budget constraints problem by constructing  many strength four testing arrays with high coverage. 
 A covering array $t$-$CA(n,k,g)$, of size $n$, strength $t$, degree $k$, and order $g$, is a $k\times n$ array on $g$ symbols such that   every $t\times n$ sub-array contains every $t\times 1$ column on $g$ symbols  at least once. It is desirable in most applications to minimise the size $n$ of covering arrays. The covering array number $t$-$CAN(k,g)$  is the smallest $n$ for which a
 $t$-$CA(n,k,g)$  exists. 
 An obvious lower bound is $$g^t\leq  t\mbox{-}CAN(k,g).$$ In this paper, we describe a construction method which is an  extension of the methods developed by  Chateauneuf, Colbourn and   Kreher \cite{chatea} and Meagher and Stevens \cite{karen}. This method improves 
 some of the best known upper bounds for strength four covering arrays with $g=3$. In the range of degrees considered in this paper, the best known results previously come from \cite{Colbourn14}; in that paper, covering arrays are also found by using a group action on the symbols (the affine or Frobenius group), but no group action on the rows is employed.  
 While for $g=3$ the group that we employ on the symbols coincides with the affine group, we accelerate and improve the search by also exploiting a group action on the rows as in \cite{chatea,karen}, and develop a search method than can be applied effectively whenever $g \geq 3$ and $g-1$ is a prime power.
 
\par There is a large literature \cite{chatea, HartmanDM} on covering arrays, and the problem of determining small covering arrays has been studied under many guises over the past thirty years. In \cite{HartmanDM}, Hartman and Raskin discussed several generalizations  motivated by their applications in the realm of software testing.   When testing a software system with $k$ parameters, each of which must be tested with $g$ values, the total number of possible test cases is $g^k$. For instance, if there are 20 parameters and three values for each parameter then the number of input combinations or test cases of this system is $3^{20}=3486784401$. 
A fundamental problem with software testing is that testing under all combinations of inputs is not feasible, even with a simple product \cite{kaner, Kuhn}. Software developers cannot test everything, but they can use combinatorial test design to identify the minimum number of tests needed to get the coverage they want.  The goal of most combinatorial testing research is to create test suites that find a large percentage of errors of a system while having a small number of tests required. Covering arrays prove useful in locating a large percentage of errors in software systems \cite{Cohen,yilmaz}. The test cases are the columns of a covering array  $t$-$CA(n,k,g)$. This is one of the five natural generalizations in \cite{HartmanDM}.

 \noindent{\it Covering arrays with budget constraints:} A practical limitation in the realm of testing is budget. 
In most software development environments, time, computing, and human resources needed to perform the testing of a component is
strictly limited. To model this situation, we consider the problem of creating  best possible test suite (covering the maximum number of 
$t$-tuples) within a fixed number of test cases. The coverage measure $\mu_t(A)$ of a testing array $A$ is defined by the ratio between the number of distinct t-tuples contained in the column vectors of  $A$ 
and the total number of $t$-tuples  given by $ {k\choose t}g^t $. 
 Our objective is to construct a testing array $A$ of size at most $n$  having largest possible coverage measure, given fixed values of $t, k, g$ and $n$. This problem is called  {\it covering arrays with  budget constraints}.\\

 We summarize the results from group theory that we use. 
Let $\mathbb{F}_q$ be a Galois field GF$(q)$ where $q=p^m$ and $p$ is prime. We adjoin to $\mathbb{F}_q$ the symbol $\infty$: it may be helpful to think of the resulting set 
 $$X=\mathbb{F}_q \cup \{\infty \}$$ as the projective  line consisting of $q+1$ points.  Recall that the projective general linear group of dimension 2 may be seen as the ``fractional linear group":
 $$ PGL(2,q)=\{\alpha~:~X \mapsto X ~|~ x\alpha=\frac{ax+b}{cx+d}, \mbox{ where  } a,b,c,d\in \mathbb{F}_q  \mbox{ and } ad-bc\neq 0\}$$ in which we define $\frac{1}{0}=\infty$, $\frac{1}{\infty}=
0$,  $1-\infty=\infty -1=\infty $, and $\frac{\infty}{\infty}=1$.  It is known that $|PGL(2,q)|=\frac{(q^2-1)(q^2-q)}{(q-1)}=(q+1)q(q-1)$ and its action on $\mathbb{F}_q\cup \{\infty\}$ is sharply 3-transitive.  For the undefined terms and more details 
see \cite[Chapter~7]{robinson}.  

 Pair-wise or 2-way interaction testing  and 3-way interaction testing are known to be effective for
different types of software testing \cite{Cohen,amity, Maity2005}.  However, software failures may be
caused by interactions of more than three parameters. A recent NIST study  indicates that failures can be triggered by interactions up
to 6 parameters \cite{Kuhn}.   Here we consider the problem of 4-way interaction testing of the parameters. 
The construction given in this paper improves many of the current best known upper bounds on 4-$CAN(k,g)$ with $g=3$ and $21\leq k\leq 74$.
This paper also presents several strength four testing arrays with high coverage measures.

\section{PGL Construction}  Let $X=GF(g-1) \cup \{\infty \}$ be the set of $g$ symbols on which we are to construct a 4-$CA(n,k,g)$.  
  We choose $g$ so that $g-1$ is a prime or prime power. \\
  
 \subsection{Case 1: Two starter vectors}
Our construction involves selecting  a group $G$ and finding  vectors $u,v\in {X}^k$, called starter vectors. We use the vectors to form a $k\times 2k$  matrix $M$.
\[ M=
\begin{pmatrix}
    u_1       & u_k  & \dots & u_2 & ~~& v_1       & v_k  & \dots & v_2\\
    u_2       & u_1  & \dots & u_3 & ~~&  v_2       & v_1  & \dots & v_3  \\
    \vdots   & \vdots&          &\vdots & ~~& \vdots   & \vdots&          &\vdots\\    
    u_{k-1}       & u_{k-2} & \dots & u_k & ~~& v_{k-1}       & v_{k-2} & \dots & v_k\\    
    u_k          &    u_{k-1} & \dots &  u_1  &~~&v_k          &    v_{k-1} & \dots &  v_1  
\end{pmatrix}.
\]
 Let $G=PGL(2,g-1)$. For each $a\in PGL(2,g-1)$, let $M^a$ be the matrix formed by the action of $a$ on the elements of $M$. The matrix obtained by 
 developing $M$ by $G$ is the $k \times 2k|G|$ matrix $M^G = [M^a:a\in G]$. Let $C$ be the $k\times g$ matrix  that 
 has a constant column with each entry equal to $x$, for each $x\in X$.  Vectors $u,v\in X^k$ are said to be  {\it starter vectors} for a $4$-$CA(n,k,g)$ if any $4\times 2k$ subarray of the   matrix $M$ has at least one representative from each non-constant orbit of $PGL(2,g-1)$ acting on 4-tuples from $X$. 
 Under this group action, there are precisely $g+11$  orbits of 4-tuples. These $g+11$ orbits are determined by the pattern of entries in their 4-tuples:
\begin{enumerate}
\item  $\{[a,a,a,a]^T: a\in X\}$
\item $\{[a,a,a,b]^T : a,b\in X, a\neq b\}$
\item  $\{[a,a,b,a]^T : a,b\in X,a\neq b\}$
\item  $\{[a,b,a,a]^T : a,b\in X, a\neq b\}$
\item $\{[b,a,a,a] ^T: a,b\in X, a\neq b\}$
\item $\{[a,a,b,b]^T : a,b\in X,a\neq b\}$
\item  $\{[a,b,a,b]^T : a,b\in X,a\neq b\}$
\item $\{[a,b,b,a]^T : a,b\in X,a\neq b\}$
 \item  $\{[a,a,b,c]^T : a,b,c\in X,a\neq b\neq c\}$
 \item  $\{[b,a,a,c]^T: a,b,c\in X,a\neq b\neq c\}$
\item  $\{[a,b,a,c]^T: a,b,c\in X,a\neq b\neq c\}$
\item  $\{[b,a,c,a]^T : a,b,c\in X,a\neq b\neq c\}$
 \item $\{[a,b,c,a]^T: a,b,c\in X,a\neq b\neq c\}$
 \item $\{[b,c,a,a]^T : a,b,c\in X,a\neq b\neq c\}$
 
 \item $g-3$ orbits of patterns with four distinct entries. The reason is this. There are $g(g-1)(g-2)(g-3)$ 4-tuples with four distinct entries and each orbit
 contains $g(g-1)(g-2)$ 4-tuples as $ |PGL(2,g-1)|=g(g-1)(g-2)$. \\
  \end{enumerate}
If starter vectors $u,v$ exist in $X^k$ (with respect to the group $G$) then there exists a 4-$CA(2kg(g-1)(g-2)+g,k,g)$. We give an example to explain the method.

\begin{example}\rm
Let $g=3$, $k=30$, $X=GF(2)\cup \{\infty \}$ and $G=PGL(2,2)$. The action of $G$ on 4-tuples from $X$ has 14 orbits:

\begin{center}
\begin{itemize}
\item [] Orb 1:  $[0000,\infty\infty\infty\infty, 1111]$
\item [] Orb 2: $[0001, 000\infty,\infty\infty\infty0,\infty\infty\infty1, 1110, 111\infty]$
\item [] Orb 3: $[1\infty\infty\infty, 1000, 0111,\infty000, 0\infty\infty\infty,\infty111]$
\item [] Orb 4: $[0100,\infty0\infty\infty, 0\infty00,\infty1\infty\infty, 1011, 1\infty11]$
\item [] Orb 5: $[11\infty1,\infty\infty1\infty, 0010, 1101, 00\infty0,\infty\infty0\infty]$
\item [] Orb 6: $[11\infty\infty,\infty\infty11, 0011, 1100, 00\infty\infty,\infty\infty00]$
\item [] Orb 7: $[\infty0\infty0, 0101,\infty1\infty1, 0\infty0\infty, 1010, 1\infty1\infty]$
\item [] Orb 8: $[\infty11\infty, 1\infty\infty1, 1001, 0110,\infty00\infty, 0\infty\infty0]$
\item [] Orb 9: $[11\infty0,\infty\infty10, 001\infty, 110\infty, 00\infty1,\infty\infty01]$
\item [] Orb 10: $[\infty0\infty1, 010\infty,\infty1\infty0, 0\infty01, 101\infty, 1\infty10]$
\item [] Orb 11: $[1\infty01, 0\infty10,\infty10\infty, 01\infty0,\infty01\infty, 10\infty1]$
\item [] Orb 12: $[1\infty0\infty, 0\infty1\infty,\infty101, 01\infty1,\infty010, 10\infty0]$
\item [] Orb 13: $[1\infty00, 0\infty11,\infty100, 01\infty\infty,\infty011, 10\infty\infty]$
\item [] Orb 14: $[1\infty\infty0, 100\infty, 011\infty,\infty001, 0\infty\infty1,\infty110]$
 
\end{itemize}
\end{center}
 The following  are starter vectors  to construct $[M^G,C]$, a 4-$CA(363,30,3)$:
$$u=(011\infty11\infty\infty\infty001\infty\infty\infty1\infty10\infty\infty0\infty1100\infty01)$$  
$$v=(11\infty\infty01101000\infty101\infty1\infty0\infty000010\infty\infty\infty).$$
 We used computer search to find $u$ and  $v$.  One can check that 
on each set of $4$ rows of $M$ there is a representative from each orbit $2-14$. 
Thus, 4-$CAN(30,3)\leq 363$. 
\end{example}

\subsection{Choice of starter vectors $u$ and $v$}   The problem is  to find two  vectors
$u,v\in X^k$ such that on each set of 
$4$ rows of $M$ there is a representative from each orbit $2-15$. 
To determine which vectors  work as starters, we define the sets $d[x,y,z]$ for positive integers $ x, y $ and $ z $ as follows:
\begin{multline*}
d[x,y,z] =\lbrace(u_{i},u_{i+x},u_{i+x+y},u_{i+x+y+z}): 0\leq i \leq k-1\rbrace \bigcup \\
\lbrace(v_{i},v_{i+x},v_{i+x+y},v_{i+x+y+z}): 0\leq i \leq k-1\rbrace\end{multline*}
where the subscripts are taken modulo $k$. 
For computational convenience, we partition the collection of $k\choose 4$ choices of four distinct rows from $k$ rows into disjoint equivalence classes.

Formally, let $S$ be the set of all $k\choose 4$ $4$-combinations of the set $\{1,2,...,k\}$. Define a binary relation $R$ on $S$ by putting

$$\{s_1,s_2,s_3,s_4\}~R~\{s_{1}^{\prime}, s_{2}^{\prime},s_{3}^{\prime},s_{4}^{\prime}\} \mbox{  iff} $$
$$\{s_1+d,s_2+d,s_3+d,s_4+d\}=\{s_{1}^{\prime}, s_{2}^{\prime},s_{3}^{\prime},s_{4}^{\prime}\} \mbox{ for some } d \in \mathbb{N}$$

\noindent where all of the addition is modulo $k$. 
Because $R$ is an equivalence relation on $S$, $S$ can be  partitioned into disjoint equivalence classes.
 The equivalence class determined by $\{s_1,s_2,s_3,s_4\}\in S$ is given by $$[\{s_1,s_2,s_3,s_4\}]=\{\{s_1+d,s_2+d,s_3+d,s_4+d\}|0\leq d \leq k-1\}.$$ 
 Without loss of generality, we may assume that $0=s_1 < s_2 <s_3 <s_4$ for each equivalence class representative $[\{s_1,s_2,s_3,s_4\}]$.  
 As an illustration, when $X=\{0,1,2,...,7\}$. $S$ is partitioned into 10 disjoint equivalence classes:\[ \begin{array}{ccccc}
 [\{0,1,2,3\}]&  [\{0,1,2,4\}]& [\{0,1,2,5\}]&  [\{0,1,2,6\}]& [\{0,1,3,4\}]  \end{array} \]
\[ \begin{array}{ccccc}
[\{0,1,3,5\}] & [\{0,1,3,6\}]& [\{0,1,4,5\}]& [\{0,1,4,6\}]& [\{0,2,4,6\}] \end{array} \]
 A distance vector $(x,y,z,w)$ is associated with every equivalence class $[\{s_1,s_2,s_3,s_4\}]$ where $x=s_2-s_1$, $y=s_3-s_2$, $z=s_4-s_3$, $w=s_1-s_4$ mod $k$. The fourth distance is redundant because $x+y+z+w=k$. 
 We rewrite the equivalence class of $4$-combinations $[\{s_1,s_2,s_3,s_4\}]$ as
$$[x,y,z]=\{i,i+x,i+x+y,i+x+y+z\}|i=0,1,2,...,k-1\}$$
For $k=8$, $[1,1,1]=[\{0,1,2,3\}]$, $[1, 1, 2]= [\{0, 1, 2, 4\}]$, $[1,1,3]=[\{0,1,2,5\}]$, $[1,1,4]=[\{0,1,2,6\}]$, $[1,2,1]=[\{0,1,3,4\}]$, $[1,2,2]=[\{0,1,3,5\}]$, $[1,2,3]=[\{0,1,3,6\}]$, $[1,3,1]=[\{0,1,4,5\}]$, $[1,3,2]=[\{0,1,4,6\}]$, $[2,2,2]=[\{0,2,4,6\}]$. 

\begin{lemma} Let $S$ be the set of all $4$-combinations of $\{1,2,3,...,k\}$. 
Then $S$ can be partitioned into disjoint equivalence classes 
$$[x,y,z]=\{i,i+x,i+x+y,i+x+y+z\}|i=0,1,2,...,k-1\}$$
where $x=1,2,...,\lfloor \frac{k}{4} \rfloor$, $y=x,x+1,...,k-1$ and $z=x,x+1,...,k-1$  such that
\begin{enumerate}
 \item [(i)] $2x+y+z < k$
 \item [(ii)] when $x=z$, $x\leq y \leq \lfloor \frac{k-2x}{2}\rfloor$
\end{enumerate}
There are no further classes distinct from these.
\end{lemma}
Before proving the result, we give an example.  When $S$ is the set of all $4$-combinations of $\{0,1,2,3,4,5,6,7\}$, $S$ can be partitioned into 10 disjoint classes: 
$[1,1,1]$, $[1, 1, 2]$, $[1,1,3]$, $[1,1,4]$, $[1,2,1]$, $[1,2,2]$, $[1,3,1]$, $[1,3,2]$ and $[2,2,2]$.\\

\noindent\begin{proof}
 Let $(x,y,z,w)$ be the distance vector corresponding to  equivalence class $[\{s_1,s_2,s_3,s_4\}]$.  
 Classes  $[\{s_1,s_2,s_3,s_4\}]$, $[x,y,z]$, $[y,z,w]$, 
 $[z,w,x]$ and $[w,x,y]$ are the same. Without loss of generality, we choose $[x,y,z]$ as class representative if $x \leq y$, $x \leq z$. Thus $1 \leq x \leq \frac{k}{4}$, $y=x,x+1,...,k-1$ and $z=x,x+1,...,k-1$.  We consider three cases: (i) $x=w$, (ii) $x=z$, (iii) $x=y$.  If $w=x$, then the classes $[x,y,z]$ and $[x,x,y]$ obtained from distance vector $(x,y,z,x)$ are the same equivalence class. The classes of the form $[x,x,y]$ are generated under case (iii) as well.  In order to avoid repetition, $w$ has to be strictly greater than $x$. That is, $w=k-x-y-z >x$ which implies $2x+y+z <k$. If $z=x$, then the classes $[x,y,z]$ and $[x,w,x]$ are the same where $y+w=k-2x$. Thus it is sufficient to consider the classes of the form $[x,y,x]$ for $y \leq \lfloor \frac{k-2x}{2} \rfloor $ only. Hence the lemma follows. 
\end{proof}

\noindent All the equivalence classes are enumerated by the following algorithm.

\begin{algorithmic}
 \STATE \textsc{Equivalence-Classes($k$)}
 \STATE \bf{Input}: $k$
\STATE \bf{Output}: All $[x,y,z]$ classes.

 \FOR{x$\gets 1$ \TO $\frac{k}{4}$ }
  \STATE{
   \FOR{y$\gets x$ \TO $k-1$ }
    \STATE{
     \IF{$y > \frac{k-2x}{2}$}
      \STATE{
       \FOR{z$\gets x+1$ \TO $k-2x-y-1$ }
        \STATE{add $[x,y,z]$}
       \ENDFOR}
      \ELSE
       \STATE{
        \IF{$y == \frac{k-2x}{2}$ and $x== \frac{k-2x}{2}$}
         \STATE{add $[\frac{k}{4}, \frac{k}{4}, \frac{k}{4}]$}
         \ELSE
          \STATE{
           \FOR{z$\gets x$ \TO $k-2x-y-1$ }
            \STATE{add $[x,y,z]$}
           \ENDFOR
        }\ENDIF
     }\ENDIF
    }\ENDFOR
   }\ENDFOR

\end{algorithmic}

\begin{theorem}
 Let $X=GF(g-1) \cup \{\infty\}$ and $G=PGL(2,g-1)$. If there exists a pair of vectors $u,v \in X^{k}$ such that each $d[x,y,z]$ has a representative from each of the orbits $2-15$, then there exists a $4$-$CA(2kg(g-1)(g-2)+g,k,g)$ covering array.
\end{theorem}
\begin{proof}
 Let $u,v \in X^k$ be vectors such that each $d[x,y,z]$ has a representation from each of the orbits $2-15$. Using $u,v$, we create the matrix $[M^G,C]$. Let $\{s_1,s_2,s_3,s_4\}$ be a member in $S$. By Lemma 1, there exists three positive integers $x_0$, $y_0$ and $z_0$ such that $\{s_1,s_2,s_3,s_4\} \in [x_0,y_0,z_0]$. It is given that $d[x_0,y_0,z_0]$ has a representative from each of the orbits 2-15. 
In other words,  if we look at the rows $s_1$, $s_2$, $s_3$, $s_4$ of $M$, we see representative from each of the $g+11$ orbits. 
 Consequently, because $PGL(2,g-1)$ is 3-transitive on $X$, $[M^G,C]$ is a $4$-$CA(2kg(g-1)(g-2)+g,k,g).$
 \end{proof}

At this stage, we make a few remarks about the size of equivalence classes defined by above choices of $ x, y$ 
and $z $.
\begin{enumerate}
\item $ k \not\equiv 0$ mod $2$ :
\\If $ k $ is an odd integer,  each class contains exactly $ k $ distinct choices from the collection of  $ k\choose 4 $ 
choices and hence there are  $ l = \frac{(k-1)(k-2)(k-3)}{24} $ distinct classes of size $ k $.
\item $ k \equiv 0 $ mod $2 $ :
\\If $ k $ is an even integer,  $ \frac{k}{2} $ can be written as sum of two positive integers $ a $ and $ b $ where $ a 
\leq b $ in $ \lfloor \frac{k}{4} \rfloor $ different ways.
\\\textit{Case 1} : If $ k \not\equiv 0 $ mod $4 $,  a class of the form $ [a,b,a] $ contains only $ \frac{k}{2} $ distinct 
choices. There are total $ \lfloor \frac{k}{4} \rfloor $ equivalence classes of the form $[a,b,a]$ with size $\frac{k}{2}$ 
and the remaining classes are of size $ k $.
\\\textit{Case 2} : If $ k \equiv 0 $ mod $4 $,  a class of the form $ [a,b,a] $ contains only $ \frac{k}{2} $ distinct 
choices and a class of the form $ [a,a,a] $  where $ a= \frac{k}{4} $ contains only $ \frac{k}{4} $ distinct choices. Here  
we get total $ \frac{k}{4} -1 $ equivalence classes of size $ \frac{k}{2} $ , exactly one class of size $ \frac{k}{4} $ and 
the remaining  classes are of size $ k $.
\end{enumerate} 

\noindent For $k=8$, there are 10 equivalence classes. The classes $[1,3,1]$ and $[2,2,2]$ are of size 4 and 4 respectively and the remaining 8 classes are of size 8 each. Thus $8\times 8 +4+2= {8\choose 4}$.  

\subsection{Case 2: Two vectors $u,v$ and a matrix $C_1$}

If we do not find vectors $u$ and $v$ such that each $d[x,y,z]$ contains a 
representative from each of the orbits $2-15$,  we look for vectors that produce an array with maximum 
possible coverage. 
In order to complete the covering conditions, we add a small matrix $C_1$. We give an example below to illustrate the technique.

\begin{example} \rm Let $k=21$ and $g=3$. Here we do not find vectors $u$ and $v$ such that  each $d[x,y,z]$ contains a 
representative from each of the orbits $2-15$.  For $k=21$, there are $285$ $[x,y,z]$ classes. All classes $[x,y,z]$ are obtained by the algorithm \textsc{Equivalence-Classes}. 
One can check that for the  vectors $$u=00001010\infty1\infty\infty10\infty\infty001\infty1$$ $$v=0000100\infty00\infty10001\infty111\infty$$  
there is a representative from each orbit $2-15$ on 276 of the $d[x,y,z]$ classes. Table \ref{orbit} shows  nine classes  which do not have representative from all the orbits: 
\begin{table}[ht] 
\caption{List of classes not having representative from all the orbits}
\label{orbit}
\begin{center}
\begin{tabular}{|c|c|}
\hline
 Class & Missing orbits  \\
\hline
 $d[1,2,2]$ & $10$ \\
 $d[1,5,6]$ & $2$ \\
 $d[1,6,12]$ & $5$ \\
 $d[1,13,5]$ & $9$ \\
 $d[2,3,8]$ & $6$ \\
 $d[2,7,3]$ & $10$ \\
 $d[2,12,3]$ & $13$ \\
 $d[3,6,8]$ & $6$ \\
 $d[3,7,7]$ & $10$ \\
 \hline
\end{tabular}
\end{center}
\end{table}

\noindent  In order to complete the covering conditions, we add a small matrix $C_1$. 
\begin{center}
$C_{1}=\left(\begin{array}{ccccccccc}
\infty&0&1&1&0&\infty&\infty&\infty&1 \\
\infty&1&1&\infty&0&0&0&1&0 \\
0&1&\infty&1&1&0&1&\infty&0 \\
0&\infty&0&0&1&0&0&0&0 \\
1&0&0&0&0&\infty&1&\infty&0 \\
\infty&0&0&\infty&0&\infty&\infty&\infty&1 \\
\infty&\infty&\infty&1&0&0&0&1&0 \\
0&1&\infty&1&\infty&0&1&1&0 \\
0&0&1&0&1&0&0&0&0 \\
1&0&0&0&1&\infty&1&\infty&0 \\
\infty&\infty&\infty&\infty&0&0&1&\infty&0 \\
\infty&1&\infty&1&0&1&\infty&\infty&\infty \\
0&1&1&0&\infty&1&\infty&1&0 \\
0&\infty&1&0&\infty&\infty&0&0&0 \\
0&0&0&\infty&1&\infty&1&0&0 \\
\infty&0&0&1&\infty&0&0&0&\infty \\
\infty&1&0&0&0&1&1&1&0 \\
1&1&1&1&0&0&1&0&0 \\
0&\infty&\infty&0&1&0&1&\infty&1 \\
0&\infty&\infty&\infty&\infty&\infty&1&0&0 \\
\infty&0&0&1&\infty&0&\infty&\infty&1\\

\end{array}\right).$
                                            
\end{center}
We use computer search to find matrix $C_1$. 
This matrix has the property that every choice of four rows in $[1,2,2]$, $[2,7,3]$ and $[3,7,7]$ contains at least one representative from orbit $10$; 
every choice of four rows in $[2,3,8]$ and $[3,6,8]$  contains at least one representative from  orbit $6$; each choice of four rows in $[1,5,6]$, $[1,6,12]$, $[1,13,5]$ and $[2,12,3]$
 contains 
at least one representative from orbit $2$, $5$, $9$ and $13$ respectively.  We also need to use 
the following matrix 
\begin{center} $C=\begin{pmatrix}0&1&\infty\\0&1&\infty\\\vdots&\vdots&\vdots\\0&1&\infty\\\end{pmatrix}$
\end{center}
\noindent
to ensure the coverage of all identical $4$-tuples. Therefore, $[M^G, ~C_1^G,~ C]$ is a 4-$CA(315,21,3)$. 
\end{example}

\subsection{Case 3: One vector $u$ and a matrix $C_1$} 
For $k=37 \mbox{~to~} 58$, we use one starter vector and a $C_1$ matrix of order $k\times \ell $ with $\ell< k$. 
Tables \ref{newresult1}, \ref{newresult2}, \ref{newresult3} and \ref{newresult4}  give a list of starter vectors and matrix $C_1$ that improves the best known bounds. When the new bound is marked with an asterisk, post-optimization has been applied (see Section \ref{sec:postop}). 

\begin{table}[ht]\small
\caption{\bf Improved strength four covering arrays for $g=3$.}
\label{newresult1}
\begin{tabular}{|l|l|l|l|}
\hline
$k $ & Starter vectors and matrix $C_1$ & New  & Old \\
        &            & bound & bound\\
\hline &&& \\
21 &\pbox{10cm}{ $u=(00001010\infty1\infty\infty10\infty\infty001\infty1)$ \\
$v=(0000100\infty00\infty10001\infty111\infty)$\\
$\arraycolsep=1.2pt\def\arraystretch{1}
C_1=\left( \begin{array}{cccccccccccccccccccccc}
\infty&\infty&0&0&1&\infty&\infty&0&0&1&\infty&\infty&0&0&0&\infty&\infty&1&0&0&\infty\\
0&1&1&\infty&0&0&\infty&1&0&0&\infty&1&1&\infty&0&0&1&1&\infty&\infty&0 \\
1&1&\infty&0&0&0&\infty&\infty&1&0&\infty&\infty&1&1&0&0&0&1&\infty&\infty&0 \\
 1&\infty&1&0&0&\infty&1&1&0&0&\infty&1&0&0&\infty&1&0&1&0&\infty&1 \\
 0&0&1&1&0&0&0&\infty&1&1&0&0&\infty&\infty&1&\infty&0&0&1&\infty&\infty \\
 \infty&0&0&0&\infty&\infty&0&0&0&\infty&0&1&1&\infty&\infty&0&1&0&0&\infty&0 \\
 \infty&0&1&0&1&\infty&0&1&0&1&1&\infty&\infty&0&1&0&1&1&1&1&\infty \\
 \infty&1&\infty&0&\infty&\infty&1&1&0&\infty&\infty&\infty&1&0&0&0&1&0&\infty&0&\infty \\
 1&0&0&0&0&1&0&0&0&0&0&\infty&0&0&0&\infty&0&0&1&0&1  \\
\end{array}
\right)^T
$ }
& 305*
      &315\\

&&&\\ 
&&&\\
22 &\pbox{20cm}{$u=(0000011\infty0\infty0110\infty1\infty\infty\infty01\infty)$ \\
$v=(00010010\infty1\infty\infty0\infty01\infty10\infty\infty1)$\\
$\arraycolsep=1.2pt\def\arraystretch{1}
C_{1}=\left(\begin{array}{cccccccccccccccccccccc}
0&\infty&\infty&0&0&0&\infty&\infty&\infty&0&0&0&\infty&\infty&0&0&0&\infty&\infty&\infty&0&0\\
\infty&\infty&0&0&0&\infty&\infty&\infty&0&0&\infty&\infty&\infty&0&0&0&\infty&\infty&\infty&0&0&\infty\\
1&\infty&1&\infty&0&\infty&0&1&\infty&1&\infty&1&\infty&1&\infty&0&0&0&1&\infty&1&0\\
0&1&1&1&0&0&1&1&\infty&0&\infty&1&1&0&0&1&1&0&0&\infty&1&\infty\\
\infty&0&0&\infty&\infty&1&0&1&\infty&0&0&\infty&\infty&\infty&0&0&1&\infty&0&0&1&\infty\\
\infty&0&\infty&1&1&1&0&1&\infty&1&\infty&0&0&1&1&1&0&0&1&1&\infty&0\\
0&0&0&\infty&\infty&1&0&0&\infty&\infty&\infty&0&\infty&0&1&0&1&0&0&0&1&\infty\\
\end{array}
\right)^T
$} &307*&315\\
&&&\\
&&&\\
27&\pbox{20cm}{$u=(1101011\infty\infty\infty0\infty00\infty\infty1\infty011\infty0100\infty)$\\
$v=(11\infty0\infty1011\infty\infty\infty0\infty0\infty01\infty00001\infty\infty\infty)$\\
$\arraycolsep=1.2pt\def\arraystretch{1}
C_{1}=\left(\begin{array}{ccccccccccccccccccccccccccc}
0&1&0&1&0&1&1&0&1&0&1&0&1&0&1&0&1&0&1&0&1&0&1&0&1&0&1\\  
0&1&0&1&0&1&0&1&0&1&0&1&0&1&0&1&1&0&1&0&1&0&1&0&1&0&1  \\
0&\infty&0&\infty&0&\infty&0&\infty&0&\infty&0&\infty&0&\infty&0&\infty&0&\infty&0&\infty&0&\infty&0&\infty&\infty&0&\infty \\ 
0&\infty&0&0&0&0&0&0&0&0&0&1&0&0&0&0&0&0&0&0&0&0&0&0&\infty&0&0        \\       
 \end{array}
\right)^T$} & 345*& 378\\     
&&&\\
\hline
\end{tabular}
\end{table}

\begin{table}[ht]\small
\caption{\bf Improved strength four covering arrays for $g=3$ (continued).}
\label{newresult2}
\begin{tabular}{|l|l|l|l|}
\hline
$k $ & Starter vectors and matrix $C_1$ & New & Old \\
           &   & bound & bound\\
\hline  
&&&\\
28&\pbox{20cm}{$u=(1\infty\infty00\infty\infty1\infty01101111\infty0\infty0101\infty\infty\infty1)$\\
$v=(\infty1011\infty110\infty000\infty1\infty\infty10\infty\infty0\infty00\infty01)$\\
$\arraycolsep=1pt\def\arraystretch{.8}
C_{1}=\left(\begin{array}{cccccccccccccccccccccccccccc}
\infty&0&\infty&0&0&\infty&0&0&\infty&0&\infty&\infty&0&\infty&\infty&0&\infty&0&0&\infty&0&0&\infty&0&\infty&\infty&0&\infty\\
\infty&0&0&1&0&1&\infty&0&\infty&1&0&\infty&\infty&0&1&0&0&\infty&0&\infty&1&0&1&\infty&0&1&1&0\\
1&0&\infty&0&\infty&\infty&0&\infty&\infty&1&\infty&1&0&0&0&1&\infty&1&0&\infty&1&\infty&1&0&\infty&0&1&1\\
0&\infty&0&\infty&0&0&0&0&\infty&0&1&0&1&0&\infty&0&1&0&\infty&0&0&1&0&0&0&0&0&0  \\              
 \end{array}
\right)^T$} &360*& 383\\              
&&&\\         
&&& \\
30&\pbox{20cm}{$u=(011\infty11\infty\infty\infty001\infty\infty\infty1\infty10\infty\infty0\infty1100\infty01)$\\
\\$v=(11\infty\infty01101000\infty101\infty1\infty0\infty000010\infty\infty\infty)$\\ }
 & 363 & 393\\  
&&&\\        
       
&&&\\
32&\pbox{20cm}{$u=(\infty1100010\infty111\infty1\infty010\infty\infty0100\infty\infty0\infty\infty010)$\\
\\$v=(\infty000\infty1\infty\infty0\infty000110\infty\infty100\infty0\infty11\infty11111)$\\}
 & 387 & 409\\ 
&&&\\
&&&\\
33&\pbox{20cm} ~~Obtained from $CA(387, 32,3)$
 & 387 & 417\\ 
&&&\\
&&&\\
34&\pbox{20cm}{$u=(00\infty101\infty\infty\infty1001\infty010\infty\infty0\infty0\infty01\infty\infty0\infty11111)$\\
\\$v=(1100\infty1\infty01\infty10110\infty\infty0\infty\infty011\infty101001\infty000)$\\}
 &410* & 423\\ 
&&&\\ 
&&&\\
35&~~Obtained from $CA(411, 34,3)$ 
 & 411& 429\\ 
&&&\\ 
&&&\\ 

35&\pbox{20cm}{$u=01\infty0\infty\infty1000\infty01\infty\infty0\infty1\infty111\infty\infty\infty01\infty01000\infty1$\\
\\$v=0\infty00111\infty0\infty110\infty11\infty110\infty010010000\infty1\infty\infty0$\\} & 423& 429\\ 
&&&\\ 

\hline
\end{tabular}
\end{table}

\begin{table}[ht]\small
\caption{\bf Improved strength four covering arrays for $g=3$ (continued).}
\label{newresult3}
\begin{tabular}{|l|l|l|l|}
\hline
$k$  &  Starter vectors and matrix $C_1$ & New & Old \\
 & & bound & bound\\
\hline
&&&\\ 

36 & ~~Obtained from $CA(423, 35,3)$
&423 & 441\\
&&&\\
&&&\\
37&\pbox{20cm}{$u=(001\infty10\infty1\infty01000\infty1100\infty101111\infty001\infty\infty\infty\infty00\infty)$\\ $C_1$: $37 \times 35$ matrix}
 & 433* & 441\\ 
&&&\\
&&&\\
39&\pbox{20cm}{$u=(001\infty\infty11\infty11\infty0001\infty11\infty101\infty\infty\infty1\infty0\infty0010\infty00\infty\infty0)$\\$C_1$: $39 \times 34$ matrix}
 & 441 & 453\\ 
&&&\\
&&&\\
41&\pbox{20cm}{$u=(\infty001\infty010\infty\infty0\infty0101111\infty\infty011\infty\infty10000\infty0\infty\infty10\infty0\infty1)$\\$C_1$: $41 \times 34$ matrix \\}
 & 453 & 465\\ 
&&&\\
42&\pbox{20cm}{$u=(\infty0111\infty1\infty\infty100\infty101\infty01000\infty011\infty1010011\infty00\infty1\infty\infty\infty)$\\$C_1$: $42 \times 35$ matrix }
 & 465 & 471\\ 
&&&\\
&&&\\
46&\pbox{20cm}{$u=(\infty00000\infty1100010\infty101\infty\infty1\infty01\infty00110\infty\infty\infty\infty11\infty1101\infty101\infty)$\\$C_1$: $46 \times 33$ matrix}
 & 477 & 483\\ 
&&&\\
&&&\\
47&\pbox{20cm}{$u=(\infty0011\infty1101\infty1\infty000\infty1\infty01\infty00\infty111010\infty00\infty\infty\infty10\infty\infty1\infty\infty1\infty\infty)$\\$C_1$: $47 \times 33$ matrix}
 & 483 & 489\\ 
&&&\\
&&&\\
48&\pbox{20cm}{$u=(01\infty\infty\infty11\infty01\infty1010111\infty\infty001\infty\infty\infty0\infty110010\infty0\infty\infty000100\infty00\infty)$\\$C_1$: $48 \times 33$ matrix}
 & 489 & 495\\ 
&&&\\
&&&\\
51&\pbox{20cm}{$u=(\infty0\infty\infty101011\infty000\infty\infty11\infty1\infty1001\infty\infty\infty\infty\infty11$\\ $\infty0\infty1\infty01111001001\infty00)$\\ $C_1$: $51 \times 32$ matrix}
 & 501 & 507\\ 
&&&\\

\hline
\end{tabular}
\end{table}

\begin{table}[ht]\small
\caption{\bf Improved strength four covering arrays for $g=3$ (continued).}
\label{newresult4}
\begin{tabular}{|l|l|l|l|}
\hline
$k$  &  Starter vectors and matrix $C_1$ & New  & Old \\
  &   & bound & bound\\
\hline

&&&\\
55&\pbox{20cm}{$u=(1\infty\infty1\infty1\infty0\infty111\infty\infty1\infty0010\infty00\infty0011011\infty1\infty0$\\ $00\infty11\infty\infty0101\infty001110\infty\infty)$\\$C_1$: $55 \times 30$ matrix}
 & 513 & 519\\ 
&&&\\
&&&\\
57&\pbox{20cm}{$u=(\infty10\infty\infty\infty0011\infty01\infty10\infty11001\infty1\infty\infty0011\infty\infty110$\\ $110111010\infty\infty1\infty0\infty0000\infty01)$\\$C_1$: $57 \times 29$ matrix }
 & 519 & 531\\ 
&&&\\
&&&\\
58&\pbox{20cm}{$u=(\infty0\infty\infty00101\infty0010\infty0\infty1\infty1000\infty0\infty11001\infty00010\infty111$\\ $\infty\infty\infty11011011\infty\infty0\infty0\infty)$\\$C_1$: $58 \times 29$ matrix }
 & 525 & 531\\ 
&&&\\
&&&\\
&&&\\
63&\pbox{20cm}{$u=(1101\infty10\infty100\infty\infty\infty00101\infty\infty0\infty0\infty\infty1\infty010\infty11\infty\infty\infty01$\\ $10\infty10110001\infty0\infty11\infty\infty0\infty0\infty11)$\\$C_1$: $63 \times 26$}
 & 537 & 549\\ 
&&&\\
&&&\\
67&\pbox{20cm}{$u=(010101\infty1100\infty100\infty11\infty\infty\infty\infty0110\infty01111\infty\infty1011\infty0\infty$ \\ $1101\infty0\infty\infty0\infty101\infty\infty1\infty\infty10000\infty00)$\\$C_1$: $67 \times 25$}
 & 555 & 561\\ 
&&&\\
&&&\\
70&\pbox{20cm}{$u=(1\infty001\infty11\infty1\infty\infty\infty0\infty11\infty0\infty0\infty1\infty00011\infty0\infty\infty\infty\infty111$\\$\infty0101001\infty010011\infty\infty010000\infty10\infty\infty1100)$\\$C_1$: $70 \times 24$}
 & 567 & 573\\ 
&&&\\
&&&\\
72&\pbox{20cm}{$u=(\infty\infty000\infty1010\infty\infty\infty\infty\infty010111000\infty11011\infty011101\infty0\infty\infty1\infty00$\\$\infty1\infty1\infty\infty010\infty101100\infty01\infty\infty\infty1\infty\infty0\infty)$\\$C_1$: $72 \times 24$}
 & 573 & 579\\ 
&&&\\
&&&\\
74&\pbox{20cm}{$u=(1\infty0010\infty\infty01\infty\infty\infty111\infty\infty1\infty\infty0100\infty\infty\infty\infty10\infty1011011\infty$\\$001100001\infty\infty0\infty0\infty0\infty\infty101100\infty1\infty01\infty111\infty)$\\$C_1$: $74 \times 24$}
 & 585 & 591\\ 
&&&\\

\hline
\end{tabular}
\end{table}

\section{Improving the  solutions}  
We examine two methods to obtain small improvements on the computational results obtained.

\subsection{Extending a solution}  

Until this point, starter vectors have been developed by applying a cyclic rotation of the starter vectors in addition to the action of PGL on the symbols.  
As in \cite{karen}, one can also consider fixing one row, and developing the remaining $k-1$ cyclically. 
This can be viewed as first finding a solution of the type already described on $k-1$ rows,  but requiring an additional property.
For the 4-subsets of $\{0,\dots,k-2\}$, equivalence classes are defined as before, with arithmetic modulo $k-1$:
$$[\{s_1,s_2,s_3,s_4\}]=\{\{s_1+d,s_2+d,s_3+d,s_4+d\}|0\leq d \leq k-2\}.$$ 
For 3-subsets $\{t_1,t_2,t_3\}$ of $\{0,\dots,k-2\}$ we define further equivalence classes as
$$[\{t_1,t_2,t_3,k-1\}]=\{\{t_1+d,t_2+d,t_3+d,k-1\}|0\leq d \leq k-2\}.$$ 
If we can place an entry  in position $k-1$ to extend the length of each starter vector so that every one of the (old and new) equivalence classes represents each of the orbits $2-15$, we obtain a 4-CA of degree $k$.

The potential advantage of this approach is that a solution for degree $k-1$ can sometimes be extended to one of degree $k$ without increasing the size of the covering array produced.  
Indeed we found that the solutions for $k-1 \in \{32,34,35\}$ do ensure that the new equivalence classes also represent each of the orbits $2-15$.  
Hence we obtain the following improvements.  
Old indicates the bound obtained by applying our methods to $k$; Improved gives the bound by applying the method to $k-1$ and ensuring that the new equivalence classes represent all orbits:
\begin{center}
\begin{tabular}{|rrr|rrr|rrr|}
\hline
$k$ & Old & Improved &$k$ & Old & Improved &$k$ & Old & Improved \\
\hline
33 & 399 & 387 & 35 & 423 & 411 & 36 &435  & 423 \\
\hline
\end{tabular}
\end{center}

\subsection{Randomized Post-optimization}\label{sec:postop}

Nayeri, Colbourn, and Konjevod \cite{nayeri} describe a post-optimization strategy which, when applied to a covering array, exploits flexibility of symbols in an attempt to reduce its size.  We applied their method to the arrays provided here, and to arrays obtained by removing one or more rows.  Because the method is described in detail elsewhere, we simply report improvements for eight values of $k$.  Basic gives the bound from starter vectors, Improved gives the bound on 4-$CAN(k,3)$ after post-optimization:
\begin{center}\small
\begin{tabular}{|rrr|rrr|rrr|}
\hline
$k$ & Basic & Improved &$k$ & Basic & Improved &$k$ & Basic & Improved  \\
\hline
19 & 309 & 300 & 20 & 309 & 303 & 21 &309  & 305 \\
22 & 309 & 307 & 27 & 351 & 345 & 28 & 363 & 360 \\
34 & 411 & 410 & 37 & 435 & 433  &&&\\
\hline
\end{tabular}
\end{center}

\section{Covering arrays with budget constraints problem}  

In this section we present several strength four testing arrays with high coverage 
measure for $g\geq 3$. The coverage measure $\mu_4(A)$ of a strength four testing array $A$ 
is defined by the ratio between the number of distinct $4$-tuples  contained in the column vectors of $A$ and the total number of 
$4$-tuples given by ${k\choose 4}g^4$. Note that the coverage measure of a covering array is always one.  
For computational convenience, we rewrite the coverage measure in terms of equivalence classes $[x,y,z]$ and $d[x,y,z]$ as follows: 
$$
\mu_4(A)=\frac{\sum\limits_{x,y,z}^{}{|[x,y,z]|\times  \mbox{number of distinct 4-tuples covered by } d[x,y,z]}}{{k \choose 4} g^{4} }.
$$
 We search by computer to find vectors $v$ with very high coverage measures. Tables \ref{coverage1} and \ref{coverage2} show vectors with high coverage, the number of test cases $(n)$ generated by our technique, and the best known size with full coverage. Comparison of our construction with best known covering array sizes shows that our construction produces significantly smaller testing arrays with very high coverage measures. 
  
 \begin{table}[ht]\footnotesize
\caption{\bf A comparison of the number of test cases $(n)$ produced by our construction with high coverage measure and best known $n$ for full coverage. For $g=5$, the elements of $GF(4)$ are represented as 0,1, 2, and 3; here 2 stands for $x$ and 3 stands for $x+1$. }
\label{coverage1}
\begin{tabular}{|c|c|c|c|c|c|}
\hline
$~~~~(g,k)~~~~$ & Vector $v$ with good coverage & Our Results & ~~Best  known  \\
              &                    & ~~$n$ ($\mu$) ~~      &  $n$~~\cite{Colbourn} \\
\hline
(3,16) & 00001001$\infty\infty$011$\infty$1$\infty$ &  99 (0.828) & 237\\
(3,17) & 0000010$\infty\infty$101$\infty$01$\infty$1  &105 (0.851) & 282\\
(3,18) & 00010$\infty$0$\infty$1001$\infty$111$\infty\infty$  &111 (0.864 ) &293 \\
(3,19) & 000010010$\infty$01$\infty$0$\infty$111$\infty$ &117 (0.883) &305 \\
(3,20) & 0000110101$\infty$0$\infty$10$\infty\infty$11$\infty$ &123 (0.892) & 314\\
(3,21) & 00001010$\infty$1$\infty\infty$10$\infty\infty$001$\infty$1 &129 (0.906) & 315 \\
(3,22) & 0000011$\infty$0$\infty$0110$\infty$1$\infty\infty\infty$01$\infty$ &135 (0.913) &315 \\
(3,23) & 0000001$\infty\infty$0101$\infty$10$\infty$10$\infty\infty\infty$1 &141 (0.923) & 315\\
(3,24) & 00000001$\infty\infty$0101$\infty$10$\infty$101$\infty\infty$1 &147 (0.924) & 315\\
(3,25) & 0000000011$\infty$0$\infty$011$\infty$01$\infty$0$\infty$11$\infty$ &153 (0.930) &363 \\
(3,28) & 1$\infty\infty$00$\infty\infty$1$\infty$01101111$\infty$0$\infty$0101$\infty\infty\infty$1 &171 (0.957) &383 \\
(3,29) &  010$\infty$00$\infty$1$\infty$0$\infty\infty\infty$101$\infty$00$\infty$000111$\infty$10 &177 (0.961) &392 \\
(3,30) & 011$\infty$11$\infty\infty\infty$001$\infty\infty\infty$1$\infty$10$\infty\infty$0$\infty$1100$\infty$01 &163 (0.969) &393 \\
(3,35) & 01$\infty$0$\infty\infty$1000$\infty$01$\infty\infty$0$\infty$1$\infty$111$\infty\infty\infty$01$\infty$01000$\infty$1 &213 (0.979) &429 \\
(3,36) & 11$\infty$0110$\infty\infty$00$\infty$111101011$\infty$001$\infty\infty\infty\infty\infty$100$\infty$0$\infty$ &219 (0.981) &441 \\
(3,38) & 1$\infty$1$\infty$111$\infty\infty$010$\infty$10$\infty\infty$00010$\infty\infty$0$\infty\infty\infty$1101$\infty\infty$100$\infty$ &231 (0.985) &447 \\
(3,39) & 001$\infty\infty$11$\infty$11$\infty$0001$\infty$11$\infty$101$\infty\infty\infty$1$\infty$0$\infty$0010$\infty$00$\infty\infty$0 &237 (0.986) &453 \\
(3,40) & 100$\infty\infty$00001$\infty\infty$1$\infty$10$\infty$000$\infty\infty\infty$0$\infty$10$\infty\infty$1$\infty$1$\infty$0111$\infty$01 &243 (0.988) &465 \\

&&&\\
  & & & \\
(4,18) & 00010021$\infty\infty\infty$21020$\infty$2 &  436 (0.851) &760 \\
(4,19) & 0000121011$\infty$01$\infty$0$\infty$221 &  460 (0.866) &760 \\
(4,20) & 0000112101202$\infty$0221$\infty$2 &  484 (0.878) & 760\\
(4,21) & 0000011021010$\infty$2$\infty$0221$\infty$ &  508 (0.887) &1012 \\
(4,22) & 0000001102$\infty$02021$\infty\infty$01$\infty$1 & 532 (0.894) &1012 \\
(4,23) & 00000001210210$\infty\infty$20112$\infty$1 &  556 (0.898) &1012\\
(4,24) & 00000000121$\infty$011$\infty$02$\infty$0$\infty$112 &  580 (0.899) &1012\\
(4,25) & 000000000121220$\infty$011$\infty$2012$\infty$ &  604 (0.901)  &1012 \\
(4,26) & 00100$\infty$2221110102$\infty$0022$\infty$020$\infty$2 &  628 (0.921)  &1012 \\
(4,27) &0100$\infty$2221110102$\infty$0022$\infty$020$\infty$2 &  652 (0.928)  &1012 \\
(4,28) &  01110$\infty$0102$\infty$021110022001$\infty$1001 &  676 (0.933)  &1012 \\
(4,29) & 0$\infty\infty$122101$\infty$000220200221220$\infty$02 &  702 (0.937)  &1012 \\
(4,30) & 10$\infty$20$\infty$020$\infty$2$\infty$2$\infty$01$\infty$2222$\infty$022002$\infty$1 &  726 (0.943)  &1012 \\
\hline
\end{tabular}
\end{table}

 \begin{table}[ht]\footnotesize
\caption{\bf A comparison of the number of test cases $(n)$ produced by our construction with high coverage measure and best known $n$ for full coverage (continued). }
\label{coverage2}
\begin{tabular}{|c|c|c|c|c|c|}
\hline
$~~~~(g,k)~~~~$ & Vector $v$ with good coverage & Our Results & ~~Best  known  \\
              &                    & ~~$n$ ($\mu$) ~~      & $n$~~ \cite{Colbourn} \\
 & & & \\

(5,21)&110131300$\infty$30010$\infty\infty$3203&  1265 (0.834) &1865\\
(5,22) & 3$\infty$32011200$\infty\infty$00$\infty$0$\infty$10010&  1325 (0.842) &1865\\
(5,23) & 0002$\infty$03100$\infty$203021332320&  1385 (0.854) &1865\\
(5,24)& 003$\infty$21022212300032302310&  1445 (0.860) &1865\\
(5,25)& $\infty$200$\infty$0$\infty\infty$31020$\infty$300303$\infty\infty$33&  1505 (0.869) &2485\\
(5,26)& 202002211000$\infty$0121031$\infty\infty$2300&  1565 (0.873) &2485\\
(5,27)& $\infty\infty$03002030$\infty$000$\infty$11$\infty$0031301$\infty$3&  1625 (0.880) &2485\\
(5,28)& 013333130320$\infty$1$\infty$1003200310300&  1685 (0.883) &2485\\
(5,29)& 00012212$\infty$010$\infty$3110031020031010&  1745 (0.891) &2485\\
(5,30)& ~~ 33001$\infty$0$\infty$000330$\infty\infty$010012$\infty$1313001~~&  1805 (0.894) &2485\\
(5,31)& 033$\infty$21333010313$\infty$303320030012020&  1865 (0.895) &2485\\
(5,32)& 310031000$\infty$330130321$\infty\infty$03031111310&  1925 (0.897) &2485\\
(5,33)& $\infty$0010$\infty\infty$3$\infty$0$\infty$2$\infty$01$\infty$00$\infty$12222$\infty\infty$03$\infty$020$\infty$&  1985 (0.904) &2485\\
(5,34)& $\infty\infty$3$\infty$00101001$\infty$0$\infty$001$\infty$002$\infty$01110231112&  2045 (0.906) &2485\\
(5,35)& 1203003303$\infty$0$\infty$013233310$\infty$032020003220&  2105 (0.906) &2485\\
(5,36)& 12022$\infty$3203230023223220001010200$\infty$2230&  2165 (0.912) &2485\\
 & & & \\

(6,25) & 000403014003033404320$\infty$1$\infty\infty$&  3006 (0.811) & 6325\\
(6,26) & $\infty$0$\infty$40021404010013010011444&  3126 (0.819) & 6456\\
(6,27)& 433$\infty\infty$01$\infty\infty$20$\infty$03020$\infty\infty$0$\infty$00401$\infty$&  3246 (0.826) &6606\\
(6,28)& 4023031100232200$\infty$21$\infty\infty$2020020&  3366 (0.829) &6714\\
(6,29)& 00$\infty$40023103301343401230334400&  3486 (0.834) &6852\\
(6,30)& 1$\infty\infty\infty$42$\infty$4040004$\infty$104$\infty$03034$\infty\infty$0300&  3606 (0.836) &6966\\
(6,31)& 44122002$\infty$2000020202031$\infty$42044001&  3726 (0.838) &7092\\
(6,32)& 44441341$\infty$424000$\infty\infty$040004410103400&  3846 (0.846) &7200\\
(6,33)& 0330344$\infty$0232133100313000030$\infty$4303$\infty$&  3966 (0.855) &7320\\

\hline
\end{tabular}
\end{table}

\section{Conclusions} In this paper, we present a construction method of strength four covering arrays with three symbols that combines an algebraic technique 
with computer search. This method improves the current best known upper bounds on 4-$CAN(k,g)$ for $21\leq k\leq 74$ and $g=3$. 
 We have also proposed a construction of strength four covering arrays with budget constraints.  In order to test 
software  with 25 parameters each having three values, our construction can generate a test suite with 153 test cases 
that ensure  
with 
 probability $0.93$ that software failure cannot be caused due to interactions of two, three or four parameters whereas the best known covering array in \cite{Colbourn}  requires 363 test cases
for full coverage.  The results show that the proposed method
could reduce the number of test cases significantly while compromising only slightly on the coverage.

\section*{Acknowledgements}
The second author gratefully acknowledges support from the Council of Scientific and Industrial Research (CSIR), India, during the work 
under CSIR senior research fellow scheme.   
The fourth author's research was supported in part by the National Science Foundation under Grant No. 1421058.

\section*{References}
\bibliographystyle{elsarticle-num}

\end{document}